\documentclass[12pt]{iopart}
\usepackage{iopams}

 \usepackage{amsthm}

\newtheorem{theorem}{\bf Theorem}
\newtheorem{lemma}{\bf Lemma}

\begin{document}
%%%% Article title %%%%
\title[An asymptotic expansion of the Casoratian]{
An asymptotic expansion of the Casorati determinant
and its application to discrete integrable systems}

%%%% Author details %%%%
\author{Masato Shinjo$^{1}$,
Masashi Iwasaki$^{2}$
and
Yoshimasa Nakamura$^{1}$}

%%%%%%%%% Insert author address %%%%%%%%
\address{$^{1}$
Graduate School of Informatics, 
Kyoto University,
Yoshidahonmachi, 
Sakyo-ku, Kyoto, 606-8501, Japan}
\address{
$^{2}$
Faculty of Life and Environmental Sciences, 
Kyoto Prefectural University,
1-5, Nakaragi-cho, Shimogamo, 
Sakyo-ku, Kyoto, 606-8522, Japan}

%%%%%%%%% Insert author E-mail %%%%%%%%
\ead{mshinjo@amp.i.kyoto-u.ac.jp}

%%%% Abstract text to be placed here %%%%%%%%%%%%
\begin{abstract}
The Hankel determinant appears in the representation 
of solutions to several integrable systems.
Asymptotic expansion of the Hankel determinant thus 
plays a key role for investigating asymptotic analysis
of such integrable system. 
In this paper, 
an asymptotic expansion formula 
of a certain Casorati determinant 
is presented as an extension of the Hankel case. 
It is also shown that
an application of it to an asymptotic analysis 
of the discrete hungry Lotka-Volterra system, which
is one of basic models in mathematical biology.
\end{abstract}
%\noindent{\it Keywords\/}:
\submitto{\JPA}
\maketitle

%%%%%%%%%%%%%%%%%%%%%%%%%%%%%%%%%%%%%%%%%%%%%%%%%%
%                               Introduction                                  % 
%%%%%%%%%%%%%%%%%%%%%%%%%%%%%%%%%%%%%%%%%%%%%%%%%%

\section{Introduction}
The Toda equation and the Lotka-Volterra (LV) system 
are the basic integrable systems 
which describe the current-voltage in  an electric circuit 
and a prey-predator relationship of distinct species, respectively.
The discrete Toda equation \cite{hirota} 
is a time-discretization of the Toda equation, 
and is known to be just equal to the recursion formula 
of the qd algorithm for computing eigenvalues 
of symmetric tridiagonal matrix \cite{rutishauser, henrici}
and singular values of bidiagonal one \cite{parlett}.
The discrete LV (dLV) system \cite{tsuji}, 
which is a time-discretization of the LV system, 
also has an interesting application to computing 
for bidiagonal singular values \cite{iwasaki}.
The solutions to both the discrete Toda equation 
and  the dLV system are expressed 
by using the Hankel determinant,
\begin{eqnarray}\label{hankel}
&H_{0}^{(n)}:=1,\nonumber\\ %\quad
&H_{j}^{(n)}:=\left|
\begin{array}{cccc}
a^{(n)}          &a^{(n+1)}   &\cdots   &a^{(n+j-1)}\\
a^{(n+1)}      &a^{(n+2)}    &\cdots   &a^{(n+j)}\\
\vdots         &\vdots      &\ddots   &\vdots\\
a^{(n+j-1)}    &a^{(n+j)}     &\cdots   &a^{(n+2j-2)}
\end{array}
\right|,\qquad
j=1,2,\dots,
\end{eqnarray}
where $j$ and $n$ correspond the discrete spatial variable 
and the discrete time one, 
respectively \cite{tsuji}.
Here, 
the formal power series $f(z)=\sum_{n=0}^{\infty}a^{(n)} z^n$ 
associated with $H_{j}^{(n)}$ 
is assumed to be analytic at $z=0$ 
and meromorphic in the disk $D=\{z| |z|<\sigma\}$.
The finite or infinite number 
of poles $u_{1}^{-1},u_{2}^{-1},\dots$ of $f(z)$
are numbered such that 
$0<|u_{1}^{-1}|<|u_{2}^{-1}|<\cdots<\sigma$.
Then, 
there exists a nonzero constant $c_j$ 
independent of $n$ such that, 
for $\varrho$ satisfying $|u_{j}|>\varrho>|u_{j+1}|$,
%%%%%%%% Hankel asymptotic %%%%%%%%
\begin{eqnarray}\label{hankelas}
H_j^{(n)}=c_j(u_1 u_2 \cdots u_j)^n 
\left(
1+\Or \left( \left( \frac{\varrho}{|u_j|}\right)^n \right) 
\right),
\end{eqnarray}
as $n\rightarrow\infty$ \cite{henrici}.
The asymptotic expansion \eref{hankelas}
of the Hankel determinant \eref{hankel} 
as $n\rightarrow\infty$ 
enables us to analyze the discrete Toda equation 
and the dLV system asymptotically 
as in \cite{rutishauser, henrici} 
and in \cite{iwasaki}, respectively.\par
A generalization of the Hankel determinant $H_j^{(n)}$ 
is the determinant of a nonsymmetric square matrix of order $j$, 
%%%%%%%% Casorati determinant %%%%%%%%
\begin{eqnarray}\label{casoratid}
&C_{k,0}^{(n)}:=1,\nonumber\\%\quad
&C_{k,j}^{(n)}:=
\left|
\begin{array}{cccc}
a_{k}^{(n)}      &   a_{k+1}^{(n)}       &\cdots   &a_{k+j-1}^{(n)}\\
a_{k}^{(n+1)}  &  a_{k+1}^{(n+1)}     &\cdots   &a_{k+j-1}^{(n+1)}\\
\vdots         &  \vdots               &\ddots   &\vdots\\
a_{k}^{(n+j-1)}&  a_{k+1}^{(n+j-1)}  &\cdots    &a_{k+j-1}^{(n+j-1)}\\
\end{array}
\right|,\nonumber\\
&\qquad k=0,1,\dots,%
\qquad j=1,2,\dots,%
\end{eqnarray}
which is called the Casorati determinant, or the Casoratian.
According to \cite{casoratibook},
the Casoratian is a useful determinant 
in the theory of difference equations, 
which plays the role similar to 
the Wronskian in the theory of differential equations, 
and it has some applications to difference equations
in mathematical physics.
No one feels something wrong that the formal power series
$f_k(z)=\sum_{n=0}^{\infty}a_k^{(n)}z^n$ is associated with the
Casorati determinant $C_{k,j}^{(n)}$ for each $k$.
The formal power series $f_0(z),f_1(z),\dots$ differ from $f(z)$ 
in that not only the subscripts but also superscripts 
appear in the coefficients.\par
The main purpose of this paper is to present 
an asymptotic expansion of $C_{k,j}^{(n)}$ 
as $n\rightarrow \infty$.
As an application of it, 
we also give an asymptotic analysis 
for the discrete hungry LV (dhLV) system,
which is a generalization of the dLV system.
The dhLV system is a time-discretization 
of the hungry LV (hLV) system
\cite{bogoyavlensky,itoh} 
which grasps more complicated prey-predator relation,
and is shown in \cite{fukuda, fukuda2, yamamoto} to enable
us to give the $LR$ and the sifted $LR$ transformations
 for computing eigenvalues 
of a banded totally nonnegative matrix 
whose all minors are nonnegative.
\par
This paper is organized as follows. 
In Section 2,
we first observe that the entries in $C_{k,j}^{(n)}$
 can be expressed by using poles of $f_k(z)$.
We next give an asymptotic expansion 
of the Casorati determinant 
in terms of  poles of $f_k(z)$ 
as $n\rightarrow\infty$ by expanding the theorem on analyticity
for Hankel determinant in \cite{henrici}.
With the help of the resulting theorem, in Section 3,
we also clarify asymptotic behaviors of the solution 
to the dhLV system.
Finally,
we give concluding remark in Section 4.

%%%%%%%%%%%%%%%%%%%%%%%%%%%%%%%%%%%%%%%%%%%%%%%%%%
%   An asymptotic expansion of the Casorati determinant    % 
%%%%%%%%%%%%%%%%%%%%%%%%%%%%%%%%%%%%%%%%%%%%%%%%%%
\section{An asymptotic expansion of the Casorati determinant }
In this section, 
we first give an expression of the entries 
of the Casorati determinant  $C_{k,j}^{(n)}$ 
in terms of the poles of  the formal power series $f_k(z)$ 
associated with $C_{k,j}^{(n)}$.
Referring to the theorem on analyticity 
for Hankel determinant in \cite{henrici},
we next present an asymptotic expansion 
of the Casorati determinant $C_{k,j}^{(n)}$ 
as $n\rightarrow\infty$ by using the poles of $f_k(z)$.
We also describe an asymptotic expansion of $C_{k,j}^{(n)}$ 
as $n\rightarrow\infty$ under some restriction 
on the poles of $f_{k}(z)$
\par
Let 
$f_k(z)=\sum_{n=0}^{\infty}a_k^{(n)} z^n$, 
which is the formal power series 
associated with $C_{k,j}^{(n)}$ for $k=0,1,\dots,$
be analytic at $z=0$ and meromorphic 
in the disk $D=\{z||z|<\sigma \}$. 
Moreover, let $r_{1,k}^{-1}, r_{2,k}^{-1},\dots,$
denote the poles of $f_k(z)$ 
such that $|r_{1,k}^{-1}|<|r_{2,k}^{-1}|<\dots<\sigma$.
By extracting the principal parts in $f_k(z)$,
we derive
\begin{eqnarray}\label{fcrk}
f_k(z)&=\frac{\alpha_{1,k}}{r_{1,k}^{-1}-z}
+\frac{\alpha_{2,k}}{r_{2,k}^{-1}-z}
+\cdots
+\frac{\alpha_{j,k}}{r_{j,k}^{-1}-z}
+\sum_{n=0}^{\infty}b_k^{(n)}z^n,
\end{eqnarray}
where $\alpha_{1,k}, \alpha_{2,k}, \dots, \alpha_{j,k}$ 
are some nonzero constants and $b_k^{(n)}$,
which contains the terms 
with respect to $r_{j+1,k}^{-1},r_{j+2,k}^{-1},\dots$, satisfies 
\begin{eqnarray}\label{bmk}
|b_k^{(n)}|\leq\mu_k\rho_k^n
\end{eqnarray}
for some nonzero positive constants $\mu_k$ and $\rho_k$ with
 $|r_{j+1,k}|<\rho_k<|r_{j,k}|$.
The proof of \eref{bmk} is given in \cite{henrici} 
through the Cauchy coefficient estimate. 
We here give a lemma for an expression of $a_k^{(n)}$ 
appearing in $f_k(z)=\sum_{n=0}^{\infty}a_k^{(n)}z^n$.
%%%%%%%%%%%%%%%%%%% Lemma1 %%%%%%%%%%%%%%%%%%%%
\begin{lemma}\label{lem}
Let us assume that the poles 
$r_{1,k}^{-1}, r_{2,k}^{-1}, \dots, r_{j,k}^{-1}$ of $f_k(z)$ 
are not multiple. 
Then $a_k^{(n)}$ is expressed 
by using $r_{1,k},r_{2,k},\dots,r_{j,k}$ as
\begin{eqnarray}\label{acrbk}
a_{k}^{(n)}=
\sum_{\ell=1}^{j} c_{\ell,k}r_{\ell,k}^{n+k+1}+b_k^{(n)},
\end{eqnarray}
where $c_{1,k},c_{2,k},\dots,c_{j,k}$ 
are some nonzero constants.
\end{lemma}
%%%%%%%%%%%%%%%%%%%% Proof %%%%%%%%%%%%%%%%%%%%%%%
\begin{proof}
The key point is the replacement 
$\alpha_{1,k}=c_{1,k} r_{1,k}^{k},\alpha_{2,k}=c_{2.k} r_{2,k}^{k},\dots,\alpha_{j,k}=c_{j,k} r_{j,k}^{k}$ 
in \eref{fcrk}, namely, 
\begin{eqnarray}\label{fcr}
f_k(z)&=
\frac{c_{1,k}r_{1,k}^k}{r_{1,k}^{-1}-z}+
\frac{c_{2,k}r_{2,k}^k}{r_{2,k}^{-1}-z}+
\dots+
\frac{c_{j,k}r_{j,k}^k}{r_{j,k}^{-1}-z}+
\sum_{n=0}^{\infty}b_k^{(n)}z^n.
\end{eqnarray} 
Since each $c_{\ell,k}r_{\ell,k}^k/(r_{\ell,k}^{-1}-z)$ in \eref{fcr} 
can be regarded as the summation of geometric series,
we get
\begin{eqnarray*}
f_k(z)&=
\sum_{n=0}^{\infty} c_{1,k}r_{1,k}^{n+k+1}z^n+
\sum_{n=0}^{\infty} c_{2,k}r_{2,k}^{n+k+1}z^n+
\cdots\\
&\qquad+
\sum_{n=0}^{\infty} c_{j,k}r_{j,k}^{n+k+1}z^n+
\sum_{n=0}^{\infty}b_k^{(n)}z^n\\
&=
\sum_{n=0}^{\infty}
\left[
\left(
\sum_{\ell=1}^{j}c_{\ell,k}r_{\ell,k}^{n+k+1}
\right)
+b_k^{(n)}
\right]z^n,
\end{eqnarray*}
which implies \eref{acrbk}.
\end{proof}
Along the line similar to 
an asymptotic expansion as $n\rightarrow\infty$ 
of the Hankel determinant $H_{j}^{(n)}$ in \cite{henrici},
we have the following theorem 
for that of the Casorati determinant $C_{k,j}^{(n)}$ 
in \eref{casoratid}.
%%%%%%%%%%%%%%%%%%% Theorem1 %%%%%%%%%%%%%%%%%%%%
\begin{theorem}\label{casoratiasgen}
Let us assume that the poles 
$r_{1,k}^{-1},r_{2,k}^{-1},\dots,r_{j,k}^{-1}$
of $f_k(z)$ are not multiple.
Then there exists
some constant $c_{\kappa_1,\kappa_2,\dots,\kappa_j}$ 
independently of $n$ such that,
as $n\rightarrow\infty$,
\begin{eqnarray}\label{casoratiasg}
C_{k,j}^{(n)}
=&\sum_{\kappa_1,\kappa_2,\dots,\kappa_j=1,2,\dots,j}
\Bigg[c_{\kappa_1,\kappa_2,\dots,\kappa_j}(r_{\kappa_1,k}r_{\kappa_2,k+1}\dots r_{\kappa_j,k+j-1})^{n}
\nonumber\\
&
\hphantom{\sum_{\kappa_1,\kappa_2,\dots,\kappa_j=1,2,\dots,j}}
\qquad\bigg(
1+
\sum_{\ell=1}^{j}\Or\left(
\left(
\frac{\rho_{k+\ell-1}}{|r_{\kappa_\ell,k+\ell-1}|}
\right)^n
\right)
\bigg)
\Bigg]
\end{eqnarray}
where 
$\rho_{k+\ell-1}$ is some constant such that
 $|r_{j+1,k+\ell-1}|<\rho_{k+\ell-1}<|r_{j,k+\ell-1}|$.
\end{theorem}
%%%%%%%%%%%%%%%%%%%%% Proof %%%%%%%%%%%%%%%%%%%%%%
\begin{proof}
By applying Lemma \ref{lem} and the addition formula of determinant 
to the Casorati determinant $C_{k,j}^{(n)}$,
we derive
\begin{eqnarray}\label{cdd}
C_{k,j}^{(n)}
=\sum_{\kappa_1,\kappa_2,\dots,\kappa_j=1,2,\dots,j}
D_{k,\kappa_1,\kappa_2,\dots,\kappa_j}^{(n)}
+\sum_{\kappa_1,\kappa_2,\dots,\kappa_j=1,2,\dots,j}
\hat{D}_{k,\kappa_1,\kappa_2,\dots,\kappa_j}^{(n)}
\end{eqnarray} 
where in the $1$st summation
\begin{eqnarray*}
D_{k,\kappa_1,\kappa_2,\dots,\kappa_j}^{(n)}:=
\left|
\begin{array}{cccc}
c_{\kappa_1,k}r_{\kappa_1,k}^{n+k+1}
&c_{\kappa_2,k+1}r_{\kappa_2,k+1}^{n+k+2}
&\dots
&c_{\kappa_j,k+j-1}r_{\kappa_j,k+j-1}^{n+k+j}\\
c_{\kappa_1,k}r_{\kappa_1,k}^{n+k+2}
&c_{\kappa_2,k+1}r_{\kappa_2,k+1}^{n+k+3}
&\dots
&c_{\kappa_j,k+j-1}r_{\kappa_j,k+j-1}^{n+k+j+1}\\
\vdots&\vdots&\ddots&\vdots
\\
c_{\kappa_1,k}r_{\kappa_1,k}^{n+k+j}
&c_{\kappa_2,k+1}r_{\kappa_2,k+1}^{n+k+j+1}
&\dots
&c_{\kappa_j,k+j-1}r_{\kappa_j,k+j-1}^{n+k+2j}
\end{array}
\right|,
\end{eqnarray*}
and $\hat{D}_{k,\kappa_1,\kappa_2,\dots,\kappa_j}^{(n)}$ 
in the $2$nd summation 
denotes a determinant
of the same form 
as $D_{k,\kappa_1,\kappa_2,\dots,\kappa_j}^{(n)}$ except that 
at least one of the $i$th columns 
are replaced with 
$\mathbf{b}_i=(b_{k+i-1}^{(n)},b_{k+i-1}^{(n+1)},\dots,b_{k+i-1}^{(n+j-1)})^{\top}$.
By evaluating the $1$st summation in \eref{cdd},
we get 
\begin{eqnarray}\label{dcr}
\fl
\sum_{\kappa_1,\kappa_2,\dots,\kappa_j=1,2,\dots,j}
D_{k,\kappa_1,\kappa_2,\dots,\kappa_j}^{(n)}
=\sum_{\kappa_1,\kappa_2,\dots,\kappa_j=1,2,\dots,j}
c_{\kappa_1,\kappa_2,\dots,\kappa_j} (r_{\kappa_1,k}r_{\kappa_2,k+1}\dots r_{\kappa_j,k+j-1})^n
\end{eqnarray}
where $c_{\kappa_1,\kappa_2,\dots,\kappa_j}$ is a constant
\begin{eqnarray}\label{ckappa}
c_{\kappa_1,\kappa_2,\dots,\kappa_j}=
\left|
\begin{array}{cccc}
c_{\kappa_1,k}r_{\kappa_1,k}^{k+1}&c_{\kappa_2,k+1}r_{\kappa_2,k+1}^{k+2}&\cdots&c_{\kappa_j,k+j-1}r_{\kappa_j,k+j-1}^{k+j}\\
c_{\kappa_1,k}r_{\kappa_1,k}^{k+2}&c_{\kappa_2,k+1}r_{\kappa_2,k+1}^{k+3}&\cdots&c_{\kappa_j,k+j-1}r_{\kappa_j,k+j-1}^{k+j+1}\\
\vdots&\vdots&\ddots&\vdots\\
c_{\kappa_1,k}r_{\kappa_1,k}^{k+j}&c_{\kappa_2,k+1}r_{\kappa_2,k+1}^{k+j+1}&\cdots&c_{\kappa_j,k+j-1}r_{\kappa_j,k+j-1}^{k+2j}
\end{array}
\right|.
\end{eqnarray}
In order to estimate the $2$nd summation in \eref{cdd},
we consider the case where 
$1$st column is replaced with $\mathbf{b}_1$.
For example,
it follows that
\begin{eqnarray*}
\fl
&\left|
\begin{array}{cccc}
b_k^{(n)}
&c_{\kappa_2,k+1}r_{\kappa_2,k+1}^{n+k+2}
&\dots
&c_{\kappa_j,k+j-1}r_{\kappa_j,k+j-1}^{n+k+j}\\
b_k^{(n+1)}
&c_{\kappa_2,k+1}r_{\kappa_2,k+1}^{n+k+3}
&\dots
&c_{\kappa_j,k+j-1}r_{\kappa_j,k+j-1}^{n+k+j+1}\\
\vdots&\vdots&\ddots&\vdots
\\
b_k^{(n+j-1)}
&c_{\kappa_2,k+1}r_{\kappa_2,k+1}^{n+k+j+1}
&\dots
&c_{\kappa_j,k+j-1}r_{\kappa_j,k+j-1}^{n+k+2j}
\end{array}
\right|
=\Or((\rho_k r_{\kappa_2,k+1} \dots r_{\kappa_j,k+j-1})^n).
\end{eqnarray*}
It is easy to check all the permutations 
for $\kappa_1,\kappa_2,\dots,\kappa_j$ 
in the $2$nd summation.
Thus, we can rewrite the second summation as
\begin{eqnarray}\label{sumorder}
\fl
\sum_{\kappa_1,\kappa_2,\dots,\kappa_j=1,2,\dots,j}
\sum_{\ell=1}^{j}
\Or(r_{\kappa_1,k}r_{\kappa_2,k+1}\dots
r_{\kappa_{\ell-1} ,k+\ell-2}
\rho_{k+\ell-1}
r_{\kappa_{\ell+1} ,k+\ell}\dots
r_{\kappa_j,k+j-1}).
\end{eqnarray}
Therefore, by taking account that 
$|r_{1,k}|>|r_{2,k}|>\dots>|r_{j,k}|>\rho_k>|r_{j+1,k}|$, 
from \eref{dcr}--\eref{sumorder}
we get \eref{casoratiasg}.
\end{proof}
Hereinafter,
let us consider the restricted case where
$r_{1,k}=r_1,r_{2,k}=r_2,\dots,r_{j,k}=r_j$ in $f_k(z)$.
This restriction admits the relationship of $a_0^{(n)},a_1^{(n)},\dots$,
for example, appearing in the next section
concerning an asymptotic analysis for dhLV system 
as $n\rightarrow\infty$.
Then, 
by the replacement of $r_{\ell,k}$ with $r_{\ell}$ in \eref{bmk}, 
we easily get
\begin{eqnarray}\label{acrb}
a_{k}^{(n)}=\sum_{\ell=1}^{j} c_{\ell,k}r_{\ell}^{n+k+1}+b_k^{(n)}.
\end{eqnarray}
As a specialization of Theorem \ref{casoratiasgen},
we thus derive the following theorem
for an asymptotic expansion of the Casorati determinant $C_{k,j}^{(n)}$
with the restricted $a_k^{(n)}$ as $n\rightarrow\infty$.
%%%%%%%%%%%%%%%%%%%% Theorem 2 %%%%%%%%%%%%%%%%%%%%
\begin{theorem}\label{casorati}
Let us assume that the poles 
$r_1^{-1}, r_{2}^{-1}, \dots, r_{j}^{-1}$ of $f_k(z)$ 
are not multiple.
Then there exists some constant $c_{k,j}\neq0$ 
independently of $n$ such that, 
for $|r_{j+1}|<\rho_k<|r_j|$, as $n\rightarrow\infty$,
\begin{eqnarray}\label{thm1}
C_{k,j}^{(n)}=c_{k,j}(r_{1} r_{2}\dots r_{j})^{n}
\left(
1+
\sum_{\ell=1}^{j}\Or
\left(
\left(
\frac{\rho_{k+\ell-1}}{|r_j|}
\right)^{n}
\right)
\right).
\end{eqnarray}
\end{theorem}
%%%%%%%%%%%%%%%%%%%%% Proof %%%%%%%%%%%%%%%%%%%%%%%
\begin{proof}
The replacement $r_{1,k}=r_{1},r_{2,k}=r_{2},\dots,r_{j,k}=r_{j}$ in \eref{ckappa} gives
\begin{eqnarray}
c_{\kappa_1,\kappa_2,\dots,\kappa_j}=
\left|
\begin{array}{cccc}
c_{\kappa_1,k}r_{\kappa_1}^{k+1}&c_{\kappa_2,k+1}r_{\kappa_2}^{k+2}&\cdots&c_{\kappa_j,k+j-1}r_{\kappa_j}^{k+j}\\
c_{\kappa_1,k}r_{\kappa_1}^{k+2}&c_{\kappa_2,k+1}r_{\kappa_2}^{k+3}&\cdots&c_{\kappa_j,k+j-1}r_{\kappa_j}^{k+j+1}\\
\vdots&\vdots&\ddots&\vdots\\
c_{\kappa_1,k}r_{\kappa_1}^{k+j}&c_{\kappa_2,k+1}r_{\kappa_2}^{k+j+1}&\cdots&c_{\kappa_j,k+j-1}r_{\kappa_j}^{k+2j}
\end{array}
\right|.
\end{eqnarray}
So, we simplify \eref{dcr} as
\begin{eqnarray}\label{dcr2}
&\sum_{\kappa_1,\kappa_2,\dots,\kappa_j=1,2,\dots,j}
D_{k,\kappa_1,\kappa_2,\dots,\kappa_j}^{(n)}\\
&\quad=
c_{\kappa_1,k}c_{\kappa_2,k+1}\dots c_{\kappa_j,k+j-1}
(r_{1}r_{2}\dots r_{j})^n\nonumber\\
&\quad\quad\quad\times\sum_{\kappa_1,\kappa_2,\dots,\kappa_j=1,2,\dots,j}
\left|
\begin{array}{cccc}
r_{\kappa_1}^{k+1}&r_{\kappa_2}^{k+2}
&\dots&r_{\kappa_j}^{k+j}\\
r_{\kappa_1}^{k+2}&r_{\kappa_2}^{k+3}
&\dots&r_{\kappa_j}^{k+j+1}\\
\vdots&\vdots&\ddots&\vdots\\
r_{\kappa_1}^{k+j}&r_{\kappa_2}^{k+j+1}
&\dots&r_{\kappa_j}^{k+2j}\\
\end{array}
\right|
\end{eqnarray}
It is of significance to note that 
there exists a constant $\rho_{k}$, 
which is not equal to one in Theorem \ref{casoratiasgen},
such that $|r_{j+1}|<\rho_k<|r_j|$.
This is because $\rho_{k}$ and $\rho_{k+1}$ 
does not always satisfy $\rho_{k}=\rho_{k+1}$ 
even if $r_{1,k}=r_{1},r_{2,k}=r_{2},\dots,r_{j,k}=r_{j}$
in Theorem \ref{casoratiasgen}.
Thus, \eref{sumorder} becomes
\begin{eqnarray}\label{sumorder2}
\sum_{\ell=1}^{j}
\Or
\left(
\left(
r_1 r_2\dots r_{j-1}
\rho_{k+\ell-1}
\right)^n
\right).
\end{eqnarray}
Therefore, from \eref{dcr2} and \eref{sumorder2} we have \eref{thm1}.
\end{proof}
Theorem \ref{casoratiasgen} is expected to be useful 
for asymptotic analysis of dynamical system 
whose solution are expressed 
in terms of the Casorati determinant $C_{k,j}^{(n)}$.
Theorem \ref{casorati} also covers an asymptotic expansion 
of the Hankel determinant.

%%%%%%%%%%%%%%%%%%%%%%%%%%%%%%%%%%%%%%%%%%%%%%%%%%
%        Asymptotic analysis for                                         % 
%              the discrete hungry Lotka-Volterra system        %
%%%%%%%%%%%%%%%%%%%%%%%%%%%%%%%%%%%%%%%%%%%%%%%%%%
\section{Asymptotic analysis for
 the discrete hungry Lotka-Volterra system}
In this section, 
we first explain that the solution to the dhLV system 
is written in terms of the Casorati determinant
By using Theorem \ref{casorati},
we next clarify an asymptotic behavior of the dhLV variables 
as $n\rightarrow\infty$.\par
The hLV system is known 
as one of the mathematical prey-predator models 
which is an extension of the LV system.
The hLV system differs from the simple LV system 
in that more than one food 
and predator exists for each species.
A skillful discretization of the hLV system enable us 
to give the dhLV system 
with hungry degree $M$,
\begin{eqnarray}\label{dhlv}
&\frac{u_{k}^{(n+1)}}{u_{k}^{(n)}}=
\prod_{j=1}^{M}
\frac{
\delta^{(n)}+u_{k+j}^{(n)}
}{
\delta^{(n+1)}+u_{k-j}^{(n+1)}
},\nonumber\\
&\qquad k=0,1,\dots,(M+1)m-M-1,
\qquad n=0,1,\dots,\\
& u_{-M}^{(n)}:=0,
u_{-M+1}^{(n)}:=0,
\dots,
u_{-1}^{(n)}:=0,
\quad\nonumber\\
&u_{(M+1)m-M}^{(n)}:=0,
u_{(M+1)m-M+1}^{(n)}:=0,
\dots,
u_{(M+1)m-1}^{(n)}:=0,
\nonumber
\end{eqnarray}
where $u_k^{(n)}$ and $\delta^{(n)}$ 
denote the number of the $k$th species 
and the discretization parameter 
at the discrete time $n$, respectively.
The dhLV system \eref{dhlv} with $M=1$ coincides 
with the simple LV system.
It is shown in \cite{tsujimoto} that 
the dhLV system \eref{dhlv} with $M=2$ is derived 
from starting the discussion on three-term recurrence 
$T_{k+1}(x)=xT_k(x)-v_k^{(n)}T_{k-2}(x)$.
Similarly, it is easy to get the arbitrary integer case of $M$ 
through considering the three-term recurrence 
$T_{k+1}(x)=xT_k(x)-v_k^{(n)}T_{k-M}(x)$.
At a first glance, 
the dhLV system \eref{dhlv} seems to differ from in \cite{fukuda} 
from the viewpoint of the position of discrete parameter $\delta^{(n)}$,
but they are essentially equivalent to each other.\par 
Let us introduce the auxiliary variable 
\begin{eqnarray}
&v_k^{(n)}:=u_{k-M}^{(n)}\prod_{j=1}^{M}(\delta^{(n)}+u_{k-M-j}^{(n)}),\nonumber\\
&\qquad k=M,M+1,\dots,(M+1)m-1 .\label{vu}
\end{eqnarray}
Then, $v_k^{(n)}$ can be expressed 
according to \cite{tsujimoto} as 
\begin{eqnarray}
v_k^{(n)}&=
\frac{\tau_{k+1}^{(n)}\tau_{k-M}^{(n)}}
{\tau_{k}^{(n)}\tau_{k-M+1}^{(n)}},
\qquad k=M,M+1\dots,(M+1)m-1,
\label{vtttt}
\end{eqnarray}
by using the tau-function $\tau_k^{(n)}$ with the following determinant
representation in the cases where $k=j(M+1)$ and $k=i+j(M+1)$, 
\begin{eqnarray*}
&\tau_0^{(n)}:=1,\qquad
\tau_{j(M+1)}^{(n)}:=\left|
\begin{array}{ccccc}
\tau_{0,M}^{(n)}&\tau_{1,M}^{(n)}&\dots&\tau_{j-1,M}^{(n)}\\
\tau_{M,M}^{(n)}&\tau_{M+1,M}^{(n)}&\dots&\tau_{M+j-1,M}^{(n)}\\
\vdots&\vdots&\ddots&\vdots\\
\tau_{(j-1)M,M}^{(n)}&\tau_{(j-1)M+1,M}^{(n)}&\dots&\tau_{(j-1)(M+1)-1,M}^{(n)}\\
\tau_{jM,i-1}^{(n)}&\tau_{jM+1,i-1}^{(n)}&\dots&\tau_{j(M+1)-1,i-1}^{(n)}
\end{array}
\right|,\\
&\tau_{i+j(M+1)}^{(n)}:=\left|
\begin{array}{ccccc}
\tau_{0,M}^{(n)}&\tau_{1,M}^{(n)}&\dots&\tau_{j-1,M}^{(n)}&\tau_{j,i-1}^{(n)}\\
\tau_{M,M}^{(n)}&\tau_{M+1,M}^{(n)}&\dots&\tau_{M+j-1,M}^{(n)}&\tau_{M+j,i-1}^{(n)}\\
\vdots&\vdots&\ddots&\vdots&\vdots\\
\tau_{(j-1)M,M}^{(n)}&\tau_{(j-1)M+1,M}^{(n)}&\dots&\tau_{(j-1)(M+1)-1,M}^{(n)}&\tau_{(j-1)(M+1),i-1}^{(n)}\\
\tau_{jM,i-1}^{(n)}&\tau_{jM+1,i-1}^{(n)}&\dots&\tau_{j(M+1)-1,i-1}^{(n)}&\tau_{j(M+1),i-1}^{(n)}
\end{array}
\right|,\\
&\qquad i=1,2,\dots M,
\end{eqnarray*}
where
\begin{eqnarray*}
&\tau_{\ell,s}^{(n)}
:=\left(
\begin{array}{cccc}
a_{\ell}^{(n)}\\
&a_{\ell+1}^{(n)}\\
&&\ddots\\
&&&a_{\ell+s}^{(n)}
\end{array}
\right)\in R^{(s+1)\times (s+1)},
\end{eqnarray*}
with the relationship concerning the evolution from $n$ to $n+1$,
\begin{eqnarray}\label{aconditiond}
a_{k}^{(n+1)}=a_{k+M}^{(n)}-(\delta^{(n)})^{M+1} a_{k}^{(n)}.
\end{eqnarray}
It is easy to check that \eref{aconditiond} admits 
the assumption
$r_{1,k}=r_1,r_{2,k}=r_2,\dots,r_{j,k}=r_j$ in $f_k(z)$.
The $1$st, the $2$nd, $\dots$, the $(j-1)$th raw and column blocks 
in $\tau_{i+j(M+1)}^{(n)}$ are $M$-by-$M$ matrices,
but the $j$th raw and column blocks in it are $(i-1)$-by-$(i-1)$ matrices.
The following lemma gives the representation of $v_k^{(n)}$ 
in terms of $C_{k,j}^{(n)}$ appearing in Section 1.
%%%%%%%%%%%%%%%%%%% Lemma2 %%%%%%%%%%%%%%%%%%%%%%
\begin{lemma}\label{vgggg}
The auxiliary variable $v_k^{(n)}$ is expressed as
\begin{eqnarray}
&v_{i+j(M+1)}^{(n)}=
\frac{\displaystyle C_{i,j+1}^{(n)}C_{i+1,j-1}^{(n)}}
{\displaystyle C_{i,j}^{(n)}C_{i+1,j}^{(n)}},\nonumber\\
&\qquad i=0,1,\dots,M-1,
\qquad
 j=1,2,\dots,m-1,\label{imj}\\
&v_{M+j(M+1)}^{(n)}=
\frac{\displaystyle C_{M,j+1}^{(n)}C_{0,j}^{(n)}}
{\displaystyle C_{M,j}^{(n)}C_{0,j+1}^{(n)}},
\qquad j=0,1,\dots.m-1.\label{mmj}
\end{eqnarray}
\end{lemma}
%%%%%%%%%%%%%%%%%%%% Proof %%%%%%%%%%%%%%%%%%%%%%%%%
\begin{proof}
Let us introduce a new determinant of a square matrix of order $j$,
\begin{eqnarray}\label{detg}
&G_{i,0}^{(n)}:=1, \nonumber\\%\quad 
&G_{i,j}^{(n)}:=\left|
\begin{array}{cccc}
a_{i}^{(n)}&a_{i+1}^{(n)}&\cdots &a_{i+j-1}^{(n)}\\
a_{i+M}^{(n)}&a_{i+M+1}^{(n)}&\cdots &a_{i+M+j-1}^{(n)}\\
\vdots&\vdots&\ddots&\vdots\\
a_{i+M(j-1)}^{(n)}&a_{i+M(j-1)+1}^{(n)}&\cdots &a_{i+(M+1)(j-1)}^{(n)}\\
\end{array}
\right|,\nonumber\\
&\qquad 
j=1,2,\dots.%
\end{eqnarray}
We give an explanation that 
$\tau_{j(M+1)}^{(n)}$ can be transformed 
into the block diagonal determinant 
with respect to $G_{0,j}^{(n)},G_{1,j}^{(n)},\dots,G_{M,j}^{(n)}$.
By interchanging 
the $2$nd, the $3$rd, $\dots$ the $j$th rows and columns 
with the $[1+(M+1)]$th, the $[1+2(M+1)]$th, 
$\dots$, the $[1+(j-1)(M+1)]$th ones in $\tau_{j(M+1)}^{(n)}$,
we observe that the same form of $G_{0,j}^{(n)}$ 
appears in the $1$st diagonal block of $\tau_{j(M+1)}^{(n)}$.
The entries in the $1$st, the $2$nd, 
$\dots$, the $j$th rows and columns 
in $\tau_{j(M+1)}^{(n)}$ are simultaneously all $0$,
except for those in the diagonal block part.%%
The permutations similar to the above 
makes the forms of $G_{1,j}^{(n)},G_{2,j}^{(n)}, \dots, G_{M,j}^{(n)}$ 
as the $2$nd, the $3$rd, $\dots$, the $(M+1)$th blocks 
in $\tau_{j(M+1)}^{(n)}$.
Thus, $\tau_{j(M+1)}^{(n)}$ can be expressed 
in terms of $G_{0,j}^{(n)},G_{1,j}^{(n)},\dots,G_{M,j}^{(n)}$ as 
\begin{eqnarray}\label{(M+1)p}
\tau_{j(M+1)}^{(n)}=\prod_{\ell=0}^{M}G_{\ell,j}^{(n)}.
\end{eqnarray}
Similarly, 
$\tau_{i+j(M+1)}^{(n)}$ can be transformed 
into the determinant of the block diagonal matrix
whose $M+1$ blocks are the notices 
in $G_{0,j+1}^{(n)},G_{1,j+1}^{(n)},\dots,G_{i-1,j+1}^{(n)}$
and
$G_{i,j}^{(n)},G_{i+1,j}^{(n)},\dots,G_{M,j}^{(n)}$.
Thus, it follows that
\begin{eqnarray}\label{(M+1)p+q}
\tau_{i+j(M+1)}^{(n)}=\left(\prod_{\ell=0}^{i-1}G_{\ell,j+1}^{(n)}\right)
\left(\prod_{\ell=i}^{M}G_{\ell,j}^{(n)}\right).
\end{eqnarray}
The cases 
where $k=i+j(M+1)$ and $k=M+j(M+1)$ in \eref{vtttt} become
\begin{eqnarray*}
&v_{i+j(M+1)}^{(n)}
=\frac{\tau_{i+j(M+1)+1}^{(n)}\tau_{i+(j-1)(M+1)+1}^{(n)}
}{\tau_{i+j(M+1)}^{(n)}\tau_{i+(j-1)(M+1)+2}^{(n)}},\\
&v_{M+j(M+1)}^{(n)}=
\frac{\tau_{(j+1)(M+1)}^{(n)}\tau_{j(M+1)}^{(n)}}{\tau_{M+j(M+1)}^{(n)}\tau_{j(M+1)+1}^{(n)}}.
\end{eqnarray*}
Consequently, by combining them with \eref{(M+1)p} and \eref{(M+1)p+q},
we get 
\begin{eqnarray}
&v_{i+j(M+1)}^{(n)}=
\frac{\displaystyle G_{i,j+1}^{(n)}G_{i+1,j-1}^{(n)}}
{\displaystyle G_{i,j}^{(n)}G_{i+1,j}^{(n)}},
\qquad i=0,1,\dots,M-1,\label{vgggg1}\\
&v_{M+j(M+1)}^{(n)}=
\frac{\displaystyle G_{M,j+1}^{(n)}G_{0,j}^{(n)}}
{\displaystyle G_{M,j}^{(n)}G_{0,j+1}^{(n)}}.\label{vgggg2}
\end{eqnarray}
\par
The entries in the $j$th row of $G_{i,j}^{(n)}$ 
are rewritten as the linear combination
$a_{i+M(j-1)+\ell}^{(n)}
=a_{i+M(j-2)+\ell}^{(n+1)}+(\delta^{(n)})^{M+1}a_{i+M(j-2)+\ell}^{(n)}$ 
for $\ell=0,1,\dots,j-1$.
By multiplying the $(j-1)$th row by $-(\delta^{(n)})^{M+1}$ and 
by adding it to the $j$th,
we get the row $(a_{i+M(j-2)}^{(n+1)},a_{i+M(j-2)+1}^{(n+1)},\dots,a_{i+(M+1)(j-2)+1}^{(n+1)})$
as the new $j$th.
Similarly, 
for the $(j-1)$th, the $(j-2)$th, $\dots$, the $2$nd rows,
it follows that
\begin{eqnarray*}
G_{i,j}^{(n)}=
\left|
\begin{array}{cccc}
a_{i}^{(n)}&a_{i+1}^{(n)}&\cdots &a_{i+j-1}^{(n)}\\
a_{i}^{(n+1)}&a_{i+1}^{(n+1)}&\cdots &a_{i+j-1}^{(n+1)}\\
\vdots&\vdots&\ddots&\vdots\\
a_{i+M(j-2)}^{(n+1)}&a_{i+M(j-2)+1}^{(n+1)}&\cdots&a_{i+(M+1)(j-2)+1}^{(n+1)}\\
\end{array}
\right|.
\end{eqnarray*}
It is here worth noting that the subscript $M$ 
can be regarded as be transformed into the superscript $1$.
Thus, $G_{i,j}^{(n)}$ in \eref{detg}
after all is equal to the Casorati determinant $C_{i,j}^{(n)}$ in \eref{casoratid}.
Therefore, 
by taking account of it in \eref{vgggg1} and \eref{vgggg2}, 
we have \eref{imj} and \eref{mmj}.
\end{proof}
Lemma \ref{vgggg} with Theorem \ref{casorati} leads to the following theorem 
for an asymptotic behavior of $v_k^{(n)}$ as $n\rightarrow \infty$.
%%%%%%%%%%%%%%%%%%% Theorem3 %%%%%%%%%%%%%%%%%%%%%%
\begin{theorem}\label{vcr}
The auxiliary variables $v_{j(M+1)}^{(n)}$, $v_{1+j(M+1)}^{(n)}$, $\dots$, $v_{M-1+j(M+1)}^{(n)}$ converge to $0$,
and $v_{M+j(M+1)}^{(n)}$ goes to some constant $\hat{c}_{j}$ as $n\rightarrow \infty$.
\end{theorem}
%%%%%%%%%%%%%%%%%%%%%% Proof %%%%%%%%%%%%%%%%%%%%%%%
 \begin{proof}
From Theorem \ref{casorati} and Lemma \ref{vgggg}, we derive, as $n\rightarrow\infty$, 
 \begin{eqnarray*}
 v_{i+j(M+1)}^{(n)}=&
\frac{c_{0,j+1}c_{i+1,j-1}}{c_{0,j} c_{i+1,j}}
\left(\frac{r_{j+1}}{r_j}\right)^n\\
&\qquad\times
\frac{\displaystyle
\left(1+\sum_{\ell=1}^{j+1}
\Or\left(\left(\frac{\rho_{\ell-1}}{|r_{j+1}|}\right)^{n}\right)\right)
\left(1+\sum_{\ell=1}^{j-1}
\Or\left(\left(\frac{\rho_{i+\ell}}{|r_{j-1}|}\right)^{n} \right)\right)}
{\displaystyle
\left(1+\sum_{\ell=1}^{j}
\Or\left(\left(\frac{\rho_{\ell-1}}{|r_{j}|}\right)^{n} \right)\right)
\left(1+\sum_{\ell=1}^{j}
\Or\left(\left(\frac{\rho_{i+\ell}}{|r_{j}|}\right)^{n} \right)\right)},\\
&\qquad\qquad i=0,1,\cdots,M-1.
 \end{eqnarray*}
Thus, by taking account that $|r_j|>|r_{j+1}|$, we see that
$v_{j(M+1)}^{(n)}\rightarrow0$, 
$v_{1+j(M+1)}^{(n)}\rightarrow0$, 
$\dots$, 
$v_{M-1+j(M+1)}^{(n)}\rightarrow0$
as $n\rightarrow\infty$.
Similarly, it follows that 
 \begin{eqnarray*}
\fl
v_{M+j(M+1)}^{(n)}=
\frac{c_{M,j+1}c_{0,j}}{c_{M,j}c_{0,j+1}}
\frac{\displaystyle
\left(1+\sum_{\ell=1}^{j+1}
\Or\left(\left(\frac{\rho_{M+\ell-1}}{|r_{j+1}|}\right)^{n} \right)\right)
\left(1+\sum_{\ell=1}^{j}
\Or\left(\left(\frac{\rho_{\ell-1}}{|r_{j}|}\right)^{n} \right)\right)}
{\displaystyle
\left(1+\sum_{\ell=1}^{j}
\Or\left(\left(\frac{\rho_{M+\ell-1}}{|r_{j}|}\right)^{n} \right)\right)
\left(1+\sum_{\ell=1}^{j+1}
\Or\left(\left(\frac{\rho_{\ell-1}}{|r_{j+1}|}\right)^{n} \right)\right)},
\end{eqnarray*}
which implies that $v_{M+j(M+1)}\rightarrow\hat{c}_j:=c_{M,j+1}c_{0,j}/(c_{M,j}c_{0,j+1})$
as $n\rightarrow \infty$.
\end{proof}
By recalling the relationship of  the dhLV variable $u_k^{(n)}$ 
to the auxiliary variable $v_k^{(n)}$ in \eref{vu},
we have the following theorem 
concerning an asymptotic convergence of $u_k^{(n)}$ as $n\rightarrow\infty$. 
%%%%%%%%%%%%%%%%%%%%% Theorem4 %%%%%%%%%%%%%%%%%%%%%
\begin{theorem}\label{dhlvconv}
The dhLV variable $u_{j(M+1)}^{(n)}$ converges to some nonzero constant $\bar{c}_{j}$, 
and $u_{1+j(M+1)-M}^{(n)}, u_{2+j(M+1)-M}^{(n)},\dots,u_{M+j(M+1)}^{(n)}$ go to $0$ 
as $n\rightarrow\infty$, provided that the limit of $\delta^{(n)}$ as $n\rightarrow\infty$ exists.
\end{theorem}
%%%%%%%%%%%%%%%%%%%%%% Proof %%%%%%%%%%%%%%%%%%%%%%%%%%
\begin{proof}
The proof is given by induction for $j$. Without loss of generality,
let us assume that $\lim_{n\rightarrow\infty}\delta^{(n)}=\delta$ 
where $\delta$ denotes some constant.
From \eref{vu}, it holds that
\begin{eqnarray}\label{vkn}
u_k^{(n)}=\frac{v_{k+M}^{(n)}}
{\displaystyle\prod_{\ell=1}^{M}(\delta^{(n)}+u_{k-\ell}^{(n)})}.
\end{eqnarray}
By taking the limit as $n\rightarrow\infty$ of the both hand side in \eref{vkn} with $k=0$ and 
by using $v_M^{(n)}\rightarrow\hat{c}_0$ as $n\rightarrow\infty$, 
we get
\begin{eqnarray}\label{u1}
 \lim_{n\rightarrow\infty}u_{0}^{(n)}
&=\bar{c}_0,
\end{eqnarray}
where $\bar{c}_0=\hat{c}_{0}/\delta^{M}$. 
By taking account of Theorem \ref{vcr} with \eref{u1} 
in the case where $k=1,2,\dots,M$ in \eref{vkn},
we successively check that 
$u_1^{(n)}\rightarrow 0, u_2^{(n)}\rightarrow 0,\dots,u_{M}^{(n)}\rightarrow 0$
as $n\rightarrow\infty$.\par
Let us assume that $u_{j(M+1)}^{(n)}\rightarrow \bar{c}_j$
 and $u_{1+j(M+1)}^{(n)}\rightarrow0, u_{2+j(M+1)}^{(n)}\rightarrow0,\dots,u_{M+j(M+1)}^{(n)}\rightarrow0$ 
as $n\rightarrow\infty$.
Equation \eref{vkn} with $k=(j+1)(M+1)$ becomes
\begin{eqnarray}\label{u3m}
u_{(j+1)(M+1)}^{(n)}&=\frac{v_{M+(j+1)(M+1)}^{(n)}}
{\displaystyle\prod_{\ell=1}^{M}(\delta^{(n)}+u_{(j+1)(M+1)-\ell}^{(n)})}.
\end{eqnarray}
It is obvious that denominator on the  right hand side of \eref{u3m}
converges to $\delta^{M}$ as $n\rightarrow\infty$ under the assumption.
By combining it with $v_{M+(j+1)(M+1)}^{(n)}\rightarrow\hat{c}_{j+1}$
as $n\rightarrow\infty$, we observe that 
$u_{(j+1)(M+1)}^{(n)}\rightarrow\bar{c}_{j+1}=\hat{c}_{j+1}/\delta^{M}$
as $n\rightarrow\infty$.
Moreover, it follows that
\begin{eqnarray}\label{u3m1}
&\lim_{n\rightarrow\infty}u_{i+(j+1)(M+1)+1}^{(n)}
=\lim_{n\rightarrow\infty}\frac{v_{i+(j+2)(M+1)}^{(n)}}
{\displaystyle\prod_{\ell=1}^{M}(\delta^{(n)}+u_{i+(j+1)(M+1)+
1-\ell}^{(n)})}=0,\nonumber\\
&\qquad i=0,1,\dots,M-1,
\end{eqnarray}
since $\prod_{\ell=1}^{M}(\delta^{(n)}+u_{i+(j+1)(M+1)+
1-\ell}^{(n)})\rightarrow\delta^{M-1}(\delta+\bar{c}_{j+1})$
and $v_{i+(j+2)(M+1)}^{(n)}\rightarrow0$ as $n\rightarrow\infty$.
\end{proof}
The convergence theorem concerning the dhLV system \eref{dhlv} in \cite{fukuda} is restricted 
in the case where 
the dhLV variable $u_k^{(n)}$ is positive
and the discretization parameter $\delta^{(n)}$ is fixed
positive for every $n$. Theorem \ref{dhlvconv} clams that
the $[j(M+1)]$th species survives 
and the $[1+j(M+1)]$th, the $[2+j(M+1)]$th, \dots, the $[M+j(M+1)]$th 
species vanish 
as $n\rightarrow\infty$ even in the case 
where $\delta^{(n)}$ is variable negative for every $n$.
Though the case of negative $u_k^{(n)}$
is not longer recognized as a biological model,
we also realize that the convergence 
is not different from the positive case.

%%%%%%%%%%%%%%%%%%%%
%  Concluding remark   %
%%%%%%%%%%%%%%%%%%%%
\section{Concluding remark}
In this paper, 
we associate a formal power series 
$f_k(z)=\sum_{n=0}^{\infty}a_k^{(n)}z^n$ 
with the Casorati determinant, 
and then give an asymptotic expansion of the Casorati determinant 
as $n\rightarrow\infty$ in Theorem\ref{casorati}.
As an application of Theorem\ref{casorati}, 
we also clarify an asymptotic behavior of the dhLV variables 
as $\rightarrow\infty$ in Theorem\ref{dhlvconv}.\par
Theorem\ref{casorati} will contribute to 
asymptotic analysis for another discrete integrable systems.
An example is the discrete hungry Toda (dhToda) equation
which is derived from the numbered box and ball system
through inverse ultra-discretization \cite{tokihiro}.
A special solution to the dhToda equation is shown
in \cite{tokihiro} to be written by using the Hankel determinant.
The solution with Casorati determinant 
is expected as more generalized solution,
since the dhToda equation has a relationship 
of variables to the dhLV system whose solution 
is given in the Casorati determinant \cite{fukuda3}.
The Casorati determinant directly appear in, for example,
the solution to the discrete Darboux-P\"oschl-Teller equation which
is a discretization of the dynamical system concerning a special
class of potentials for 
1-dimensional Schr\"odinger equation \cite{gaillard}.\par
It is proved in \cite{fukuda2} that the dhLV system \eref{dhlv} 
with fixed positive $\delta^{(n)}$ is associated 
with the $LR$ transformation 
for a totally nonnegative matrix.
The paper \cite{yamamoto} also suggests that 
the dhLV system \eref{dhlv} 
with variable negative $\delta^{(n)}$ 
generates the sifted $LR$ transformation.
Theorem\ref{dhlvconv} will be useful 
for investigating the convergence of the shifted $LR$ transformation 
based on the dhLV system \eref{dhlv}
in the variable negative case of $\delta^{(n)}$.\par
%%%%%%%%%%Acknowledgment%%%%%%%%%%%%%%%%%%%%
%\ack

%%%%%%%%%% Insert bibliography here %%%%%%%%%%%%%%

\end{document}